\documentclass{sl}

\usepackage{slsec}     %% PACKAGES for Studia Logica
\usepackage{stud_log} 
\usepackage{slfoot}
\usepackage{slthm}     %% Theorem-like definitions as in amsthm,
\usepackage{amssymb,amsmath}
\usepackage{graphicx,psfrag}

\newcommand{\Ph}{\mathsf{Ph}}
\newcommand{\Q}{\mathsf{Q}}
\newcommand{\B}{\mathsf{B}}
\newcommand{\W}{\mathsf{W}}
\newcommand{\ev}{\mathsf{ev}}
\newcommand{\wl}{\mathsf{wl}}
\newcommand{\IOb}{\mathsf{IOb}}
\renewcommand{\time}{\mathsf{time}}
\newcommand{\dist}{\mathsf{dist}}
\newcommand{\leteq}{\,\mbox{$:=$}\,}
\newcommand{\then}{\rightarrow}
\newcommand{\Then}{\Longrightarrow}
\renewcommand{\iff}{\leftrightarrow}
\newcommand{\Iff}{\enskip \Longleftrightarrow\ }

\newcommand{\setclose}{\}}
\newcommand{\Setclose}{\,\right\}}
\newcommand{\setmid}{\,:\,}
\newcommand{\setopen}{\{}
\newcommand{\Setopen}{\left\{\,}
\renewcommand{\d}{\mathit{d}}
\newcommand{\convex}{\mathsf{Conv}}
\newcommand{\concave}{\mathsf{Conc}}
\newcommand{\Bw}{\mathsf{Bw}}
\newcommand{\mtwp}{\mathsf{meetTwP}}

\usepackage{color} 
\definecolor{thmcolor}{rgb}{0,0,.4} 
\definecolor{remarkcolor}{rgb}{0,.2,0} 
\definecolor{proofcolor}{rgb}{.4,0,0} 
\definecolor{quecolor}{rgb}{.2,.2,0} 
\definecolor{axcolor}{rgb}{.23,0,.23}
\definecolor{axbgcolor}{rgb}{1,.6,1} 
\definecolor{defbgcolor}{rgb}{0.9,0.9,0.6} 
\definecolor{thmbgcolor}{rgb}{0.8,0.8,1} 
\definecolor{rmbgcolor}{rgb}{0.7,1,0.7} 
\definecolor{proofbgcolor}{rgb}{1,0.7,0.7} 
 
\renewcommand{\proofname}{\bf Proof:}

\newcommand{\ax}[1]{\textcolor{axcolor}{\ensuremath{\mathsf{#1}}}} 
\newcommand{\Ax}[1]{\textcolor{axcolor}{\colorbox{axbgcolor}{\ensuremath{\mathsf{#1}}}}} 
\newcommand{\df}[1]{{\bf #1}} 

\newcommand{\Df}[1]{\setlength{\fboxsep}{2pt}\colorbox{defbgcolor}{\ensuremath{#1}}\setlength{\fboxsep}{3pt}} 
\newcommand{\Dff}[1]{\setlength{\fboxsep}{0pt}\colorbox{defbgcolor}{\ensuremath{#1}}\setlength{\fboxsep}{3pt}}

\theoremstyle{definition}  
\newtheorem{thm}{\bf Theorem}
\newtheorem{cor}{\bf Corollary}
\newtheorem{lem}{\bf Lemma}
\newtheorem{prop}{\bf Proposition}
\newtheorem{conv}{\bf Convention}
\newtheorem{que}{\bf Question}
\newtheorem{rem}{\bf Remark}

\begin{document}

%\titlerunning{Geom. Char. of the Twin Paradox and its Variants}
\title{A Geometrical Characterization of the Twin Paradox and its Variants}

\author{Gergely Sz\'ekely}
%\footnote{Alfr\'ed R\'enyi Institute of Mathematics, Hungarian Academy of Sciences,
% Re\'altanoda utca 13-15, H-1053, Budapest, Hungary}} 
%\email{turms@renyi.hu}}

\maketitle

\begin{abstract}
The aim of this paper is to provide a logic-based conceptual analysis
of the twin paradox (TwP) theorem within a first-order logic framework.  A
geometrical characterization of TwP and its variants is given.  It is
shown that TwP is not logically equivalent to the assumption of
the slowing down of moving clocks, and the lack of TwP is not logically
equivalent to the Newtonian assumption of absolute time.
The logical connection between TwP and a symmetry axiom of special
relativity is also studied.
\end{abstract}

\Keywords{twin paradox; geometrical characterization; logical
  foundations; axiomatization; special relativity}

%%%%%%%%%%%%%%%%%%%%%%%%%%%%%%%%%%%%%%%%%%
\section{Introduction}%
%%%%%%%%%%%%%%%%%%%%%%%%%%%%%%%%%%%%%%%%%

The twin paradox (TwP) theorem is one of the most famous predictions of
special relativity.  According to TwP, if a twin makes a journey into
space, he will return to find that he has aged less than
his twin brother who stayed at home.  However
surprising TwP is, it is not a contradiction. It is only a fact that
shows that the concept of time is not as simple as it seems to
be.\footnote{Unfortunately, it is still not uncommon for people
  who misinterpret the word `paradox' to try to find contradictions in
  relativity theory, that is why we think it important to note here
  that its original meaning is ``a statement that is seemingly
  contradictory and yet is actually true,'' that is, it has nothing to
  do with logical contradiction.  With the nearly century long
  fruitless debate in view, perhaps it would be better to call the
  paradoxes of relativity theory simply effects, thus saying ``twin
  effect'' instead of ``twin paradox,'' but for the time being it
  appears to be a hopeless effort to have this idea generally
  accepted.  Anyway, we would like to emphasize that it is absolutely
  pointless to try to find a logical contradiction in relativity
  theory, as its consistency has been proved, see \cite[Corollary 11.12,
    p.644]{logst}, \cite[p.77]{pezsgo}.}

A more optimistic consequence of TwP is the following. Suppose you
would like to visit a distant galaxy 200 light years away. You are
told it is impossible because even light travels there for 200
years. But you do not despair, you accelerate your spaceship nearly to
the speed of light. Then you travel there in 1 year of your
(proper) time. You study there whatever you wanted, and you come
back in 1 year. When you arrive back, you aged only 2 years. So you
are happy, but of course you cannot tell the story to your brother,
who stayed on Earth. Alas you can tell it to your
grand-\ldots-grand-children only.  In this way TwP also makes time
travel to the future possible.

In this paper we use the axiomatic method to provide a logic-based
conceptual analysis of the TwP theorem. We work within the first-order logic
(FOL) framework of \cite{pezsgo}, \cite{AMNsamples}, \cite{logst}.  We
logically compare TwP and a prediction (slowing down of moving clocks)
as well as a symmetry axiom of special relativity.  This analysis is
based on our geometrical characterization of TwP, see Theorem
\ref{thm-twp}.  We show that TwP is logically weaker than the
assumption of the slowing down of moving clocks, see Theorem
\ref{thm-slowtime}. We also show that TwP is logically weaker than a
symmetry axiom of special relativity, see Theorem \ref{thm-simdist}.
Since we prove our geometrical characterization in a general
kinematics setting, we can use it to derive consequences on Newtonian
kinematics too. We show that the absoluteness of time (in the
Newtonian sense) is not equivalent to the lack of the twin paradox
(No-TwP) without assuming a strong theoretical axiom, see Theorem
\ref{thm-univtime}.

Why is it useful to apply the axiomatic method to relativity theory?
For one thing, this method makes it possible  to understand the
role of any particular axiom. We can check what happens to our theory
if we drop, weaken or replace an axiom. For instance, it has been
shown by this method that the impossibility of faster than light
motion is not independent from other assumptions of special
relativity, see \cite[\S 3.4]{pezsgo}, \cite{AMNsamples}.  More
boldly: it is superfluous as an axiom because it is provable as a
theorem from much simpler and more convincing basic assumptions.  The
linearity of transformations between inertial observers (inertial reference frames)
can also be proven from some plausible assumptions, therefore it need
not be assumed as an axiom, see \cite{pezsgo},
\cite{AMNsamples}. 

The usual approaches to special theory of relativity base the theory
on two postulates, namely, Einstein's principle of relativity and that
the velocity of light is independent of its source.  Some authors give
a mathematical argument to prove that Einstein's principle of
relativity implies the second postulate, see, e.g., \cite{fock},
\cite{sfarti}. However, these approaches contain several tacit
assumptions besides the named postulates. So from the point of view of
axiomatic foundations of relativity theory, they are not explicit
enough. In an adequate axiomatic foundational work it is desirable to
state every assumption explicitly.\footnote{The logical formulation of
  Einstein's principle of relativity is not an easy task since it is
  difficult to capture axiomatically what ``the laws of nature''
  are. Therefore we will use a different approach here. For details on
  the axiomatic reformulation of Einstein's principle of relativity,
  see \cite{pezsgo}, \cite[\S 2.8.3]{Mphd}.}

Getting rid of unnecessary axioms of a physical theory is important
because we do not know whether an axiom is true or not, we just assume
so. We can only be sure of outcomes of concrete
experiments but they rather correspond to
(existentially quantified) theorems and not to axioms. In the
literature it is common to use the term ``empirical fact'' for
universal generalization of an empirical fact (elevated to the level
of axioms), see, e.g., \cite[\S 4]{GSz}, \cite{Szabo}. However, because of
their falsifiability it would be better to call them empirical axioms
(postulates based on outcomes of concrete experiments).

Similarly, if we axiomatize a theory, we can
ask which axioms are responsible for a certain prediction of the
theory. This kind of reverse thinking helps to answer the
why-type questions of relativity. For example, we can take the twin
paradox and check which axiom of special relativity was and which one
was not needed to derive it. The weaker an axiom system is, the better
answer it offers to the question: ``Why is the twin paradox
true?''. For details on answering why-type questions of
relativity by the methodology of the present work, see \cite{wqp}. For
further reasons why to apply the axiomatic method to spacetime
theories, see, e.g., \cite{pezsgo}, \cite{AMNsamples}, \cite{guts},
\cite{schutz}, \cite{suppes}.

Applying mathematical logic in foundations of relativity theories is
not a new idea at all. It goes back to such leading mathematicians and
philosophers as Hilbert, Reichenbach, Carnap, G{\"o}del, Tarski,
Suppes and Friedman, among others. 
The work of our school of Logic and
Relativity led by Andr{\'e}ka and N{\'e}meti is continuation to
their research.
In a spirit similar to ours, there is a large variety of works devoted
to logical axiomatizations of relativity, see, e.g., Ax~\cite{ax},
Benda~\cite{benda}, Goldblatt~\cite{goldblatt},
Mundy~\cite{mundy-oaomstg}, \cite{mundy-tpcomg},
Pambuccian~\cite{pambuccian}, Robb~\cite{Robb},
Suppes~\cite{suppes-sopitposat}, Schutz~\cite{schutz},
\cite{schutz-aasfmst}, \cite{Schu}.

Our general aims are to axiomatize relativity theories within pure FOL
using simple, comprehensible and transparent basic assumptions
(axioms) only; to prove the surprising predictions (theorems) of
relativity theories from a minimal number of convincing axioms; to
eliminate tacit assumptions from relativity by replacing them with
explicit axioms formulated in FOL (in the spirit of the FOL foundation
of mathematics and Tarski's axiomatization of geometry); and to
provide a foundation for relativity theory similar to that of
mathematics, cf.\ Hilbert's 6th problem \cite{corry}.  In our
perspective axiomatization is only a first step to logical and conceptual analysis
where the real fun begins.

For good reasons, the foundation of mathematics was performed strictly
within FOL.  A reason for this fact is that staying within
FOL helps to avoid tacit assumptions. Another reason is
that FOL has a complete inference system while
second-order logic (and thus any higher-order logic) cannot have
one. For further reasons why to stay within FOL when
dealing with axiomatic foundations, see, e.g., \cite[\S Appendix: Why
  FOL?]{pezsgo}, \cite{ax}, \cite{vaananen}, \cite{wolenski}.

%%%%%%%%%%%%%%%%%%%%%%%%%%%%%%%%%%%%%%%%%%%%%%%%%%%%%%
\section{A FOL axiom system of kinematics}%
%%%%%%%%%%%%%%%%%%%%%%%%%%%%%%%%%%%%%%%%%%%%%%%%%%%%%%
\label{ax-sec}

Here we explain our basic concepts. We deal with kinematics, i.e.,
with the motion of \textit{bodies} (anything which can move, e.g.,
test-particles, reference frames, electromagnetic waves or centers of
mass).  We represent motion as the changing of spatial location in
time.  Thus we use reference frames for coordinatizing events
(meetings of bodies).  {\it Quantities} are used for marking time and
space.  The structure of quantities is assumed to be an {\it ordered
  field} in place of the field of real numbers.\footnote{ Using
  ordered fields in place of the field of real numbers increases the
  flexibility of the theory and minimizes the amount of mathematical
  presuppositions.  For further motivation in this direction, see,
  e.g., Ax~\cite{ax}.  Similar remarks apply to our other
  flexibility-oriented decisions, e.g., to keep the dimension of
  spacetime as a variable.}  For simplicity, we associate reference
frames with special bodies which we call {\it observers}.\footnote{The
  body associated to a reference frame is nothing else than a label on
  the reference frame making it easier to talk about its motion.}
  Observations are formulated by means of the {\it worldview
    relation}.

There are several reasons for using observers (or coordinate systems,
or reference frames) instead of a single observer-independent
spacetime structure.  One is that it helps to weed out unnecessary
axioms from our theories.  Nevertheless, we state and emphasize the
logical equivalence\footnote{By logical equivalence, we mean
  definitional equivalence.} of observer-oriented and
observer-independent approaches to relativity theory, see, e.g.,
\cite[\S 4.5]{Mphd}.

Keeping the foregoing in mind, let us now set up the FOL language of
our axiom systems.  First we fix a natural number $\Df{d}\ge 2$ for
the dimension of spacetime.  We use a two-sorted language: $\Df{\B}$
is the sort of (potential) {\bf bodies} and $\Df{\Q}$ is the sort of {\bf
  quantities}.  Our language contains the following non-logical
symbols:
%\begin{itemize}
%\item 
unary relation symbol \Df{\IOb} (inertial \df{observers});
%\item 
binary function symbols \Df{+}, \Df{\cdot} and a binary relation
  symbol \Df{<} (the field operations and the ordering on $\Q$); and
%\item
 a $2+d$-ary relation symbol \Df{\W} (\df{worldview relation}).
%\end{itemize}

The variables of sort $\B$ are denoted by $m$, $k$, $a$, $b$ and $c$; 
and those of sort $\Q$ are denoted by $p$, $q$, $r$, $x$ and $y$.
$\IOb(m)$ is translated as ``$m$ is an (inertial) observer.''  We use the
worldview relation $\W$ to speak about coordinatization by translating
$\W(m,b,x_1,\ldots, x_d)$ as ``observer $m$ coordinatizes body $b$ at
spacetime location $\langle x_1,\ldots,x_d\rangle$,'' that is, at space
location $\langle x_2,\ldots,x_d\rangle$ at instant $x_1$.

Body terms are just the variables of sort $\B$.  Quantity terms are
the variables of sort $\Q$ and what can be built up from quantity
terms by using the field operations ($+,\cdot\,$).  $\IOb(m)$, $\W(m,b,x_1,\ldots,
x_d)$, $m=b$, $x_1=x_2$ and $x_1<x_2$ are the so-called {atomic
  formulas} of our FOL language, where $m,b,x_1,\dots,x_d$ can
be arbitrary terms of the required sorts.  The \df{formulas} of our
FOL language are built up from these atomic formulas by using
the logical connectives {\it not} (\Df{\lnot}), {\it and}
(\Df{\land}), {\it or} (\Df{\lor}), {\it implies}
(\Df{\rightarrow}), {\it if-and-only-if}
(\Df{\leftrightarrow}) and the quantifiers {\it exists} $x$
(\Df{\exists x}) and {\it for all $x$} (\Df{\forall x}) for every
variable $x$.  To abbreviate formulas of FOL we often
omit parentheses according to the following convention.  Quantifiers
bind as long as they can, and $\land$ binds stronger than
$\rightarrow$.  For example, $\forall x\enskip
\varphi\land\psi\rightarrow\exists y\enskip \delta\land\eta$ means
$\forall x\big((\varphi\land\psi)\rightarrow\exists
y(\delta\land\eta)\big)$.

We use first-order set theory as a meta theory to speak about model
theoretical terms, such as models, validity, etc.  The \df{models} of this
language are of the form
\begin{equation*}
\Df{\mathfrak{M}} = \langle \B, \Q;\IOb_\mathfrak{M},+_\mathfrak{M},\cdot_\mathfrak{M},<_\mathfrak{M},\W_\mathfrak{M}\rangle,
\end{equation*}
where $\B$ and $\Q$ are nonempty sets and $\IOb_\mathfrak{M}$ is a
unary relation on $\B$, $+_\mathfrak{M}$ and $\cdot_\mathfrak{M}$ are
binary functions and $<_\mathfrak{M}$ is a binary relation on $\Q$,
and $\W_\mathfrak{M}$ is a relation on $\B\times \B\times
\Q\times\dots\times \Q$.  Formulas are interpreted in $\mathfrak{M}$
in the usual way.

We formulate each axiom at two levels.  First we give an intuitive
formulation, then a precise formalization using our logical notation
(which can easily be translated into FOL formulas by
inserting the FOL definitions into the formalizations).
We seek to formulate easily understandable axioms in FOL.

We use the notation $\Df{\Q^n}\leteq\Q\times\ldots\times \Q$
($n$-times) for the set of all $n$-tuples of elements of $\Q$.  If
$p\in \Q^n$, we assume that $p=\langle p_1,\ldots,p_n\rangle$, i.e., $p_i\in\Q$ denotes the $i$-th component of the $n$-tuple $p$.
Specially, we write $\W(m,b,p)$ in place of $\W(m,b,p_1,\dots,p_d)$,
and we write $\forall p$ in place of $\forall p_1\dots\forall p_d$,
etc.  To abbreviate formulas, we also use bounded quantifiers in the
following way: $\forall x\, \varphi(x)\rightarrow \psi$ and $\exists x\,
\varphi(x)\land \psi$ are abbreviated to $\forall
x\in\varphi\enskip \psi$ and $\exists x\in\varphi\enskip \psi$,
respectively. For example, we write
\begin{equation*}\forall m\in \IOb\;\exists
b\in\B\;\exists p\in\Q^d\quad W(m,b,p)
\end{equation*}
instead of 
\begin{equation*}\forall
m\;\IOb(m)\rightarrow
 \exists b\;\B(b)\land  \exists p \enskip \Q(p_1)\land\ldots\land\Q(p_d)\land W(m,b,p)
\end{equation*}
to formulate that every observer observes a body somewhere.

To be able to add, multiply and compare measurements by observers, we
provide an algebraic structure for the set of quantities by our first axiom.
\begin{description}
\item[\Ax{AxEOF}]
The \df{quantity part} $\left< \Q;+,\cdot, < \right>$ is a Euclidean ordered field
(i.e., a linearly ordered field in which positive
elements have square roots).
\end{description}
For the FOL definition of linearly ordered field, see, e.g.,
\cite{chang-keisler}.  We use the usual field operations \Dff{0, 1, -,
  /,\sqrt{\phantom{i}}} definable within FOL.  We also use the
vector-space structure of $\Q^n$, i.e., if $p,q\in \Q^n$ and
$\lambda\in \Q$, then $\Df{p+q, -p, \lambda\cdot p}\in \Q^n$; the
\df{length} of $p\in \Q^n$ is defined as
\begin{equation*}
\Df{|p|}\leteq \sqrt{p_1^2+\ldots+p_n^2}
\end{equation*} 
for any $n\ge 1$, and $\Df{o}\,\leteq\langle 0,\ldots,0\rangle$
denotes the \df{origin}. The set of positive elements of $\Q$ (i.e.,
the set $\setopen x\in \Q:0<x\setclose$) is denoted by $\Df{\Q^+}$.

\begin{figure}[h!btp]
\small
\begin{center}
\psfrag{mm}[bl][bl]{$wl_m(m)$}
\psfrag{mk}[br][br]{$wl_m(k)$}
\psfrag{mb}[t][t]{$wl_m(b)$}
\psfrag{mph}[tl][tl]{$wl_m(ph)$}
\psfrag{kk}[br][br]{$wl_k(k)$}
\psfrag{km}[b][bl]{$wl_k(m)$}
\psfrag{kb}[t][t]{$wl_k(b)$}
\psfrag{kph}[tl][tl]{$wl_k(ph)$}
\psfrag{p}[r][r]{$p$}
\psfrag{k}[bl][bl]{$k$}
\psfrag{b}[tl][tl]{$b$}
\psfrag{m}[bl][bl]{$m$}
\psfrag{ph}[tl][tl]{$ph$}
\psfrag{evm}[tl][tl]{$\ev_m$}
\psfrag{evk}[bl][bl]{$\ev_k$}
\psfrag{Evm}[r][r]{$Ev_m$}
\psfrag{Evk}[r][r]{$Ev_k$}
\psfrag{Ev}[r][r]{$Ev$}
\psfrag{T}[tl][tl]{$\ev_m(p)=\ev_k(q)$}
\psfrag{q}[r][r]{$q$}
\psfrag{Cdk}[l][l]{$Cd_k$}
\psfrag{Cdm}[l][l]{$Cd_m$}
\psfrag{Crdk}[l][l]{$Crd_k$}
\psfrag{Crdm}[l][l]{$Crd_m$}
\psfrag{t}[lb][lb]{$$}
\psfrag{o}[t][t]{$o$}
\psfrag{fkm}[t][t]{$w^k_m$}
\psfrag*{text1}[cb][cb]{worldview of $k$}
\psfrag*{text2}[cb][cb]{worldview of $m$}
\includegraphics[keepaspectratio, width=0.82\textwidth]{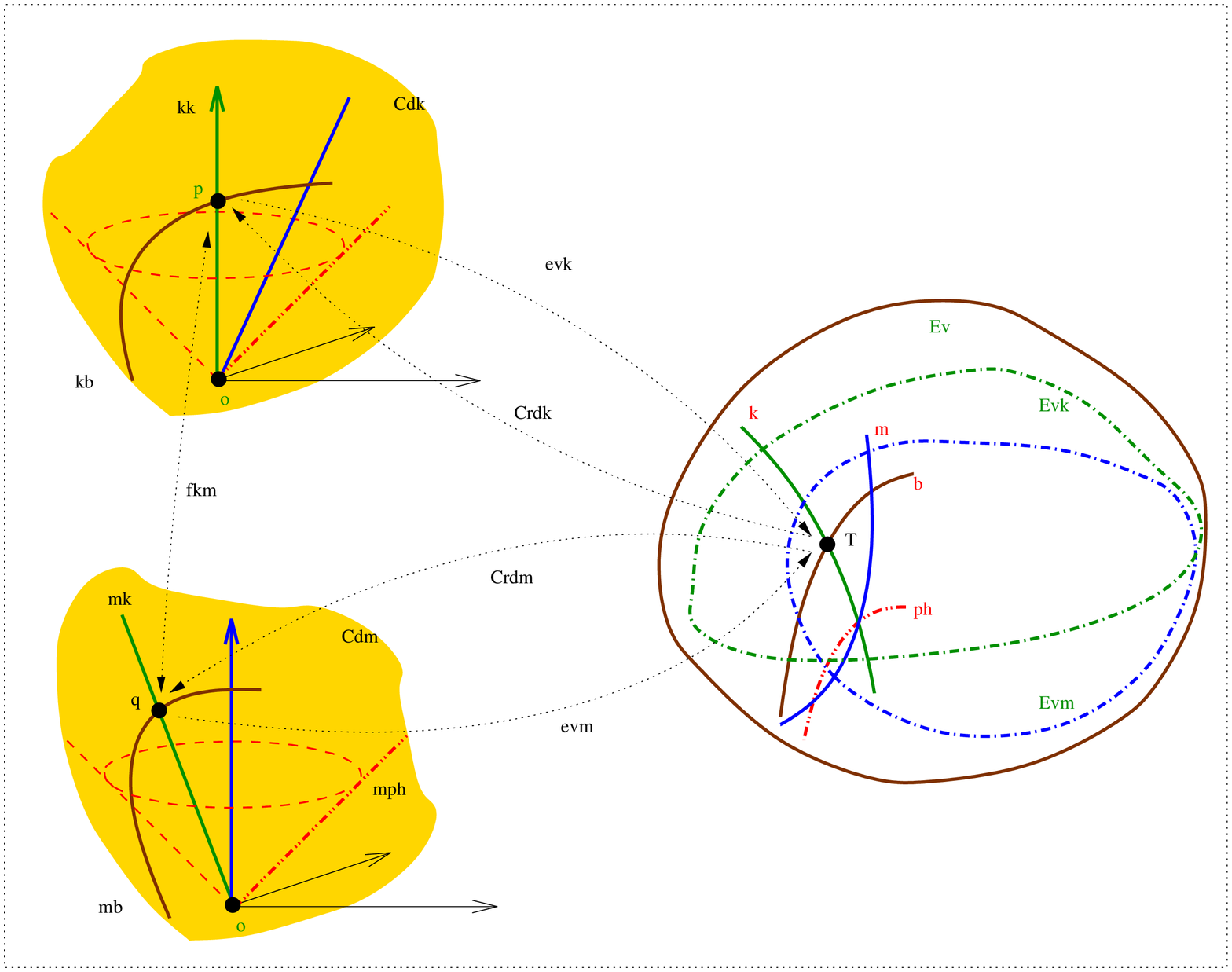}
\caption{\label{fig-fmk} Illustration of the basic definitions}
\end{center}
\end{figure}

We need some definitions and notations to formulate our other axioms.
The set $\Q^d$ is called the \df{coordinate system} and its elements are referred to as \df{
coordinate points}. 
We use the notations 
\begin{equation*}
\Df{p_\sigma}\leteq\langle p_2,\ldots, p_d\rangle \quad \text{ and }\quad \Df{p_\tau}\leteq p_1
\end{equation*} 
for the \df{space component} and the \df{time component} of $p\in\Q^d$, respectively.

Our first axiom on observers simply states that each observer thinks
that it is stationary in the origin of the space part of its
coordinate system.

\begin{description}
\item[\Ax{AxSelf}] An observer observes itself at a coordinate point
  iff the space component of this point is the origin:
\begin{equation*}
\forall m \in \IOb \enskip \forall p\in \Q^d\quad W(m,m,p) \iff p_\sigma=o.
\end{equation*}
\end{description}

The \df{event} (the set of bodies) observed by observer $m$ at
coordinate point $p$ is denoted by $\ev_m(p)$, i.e.,
\begin{equation*}
\Df{\ev_m(p)}\leteq\Setopen b\in \B \::\: \W(m,b,p)\Setclose,
\end{equation*}
and the \df{event-function} of $m$ is the function that maps coordinate point $p$ to event $\ev_m(p)$.
Let $Ev_m$ denote the set of nonempty events coordinatized by observer $m$, i.e.,
\begin{equation*}
\Df{Ev_m}\leteq\Setopen \ev_m(p) \::\: \ev_m(p)\neq\emptyset\Setclose,
\end{equation*}
and $Ev$ denote the set of all observed events, i.e.,
\begin{equation*}
\Df{Ev}\leteq\Setopen e\in Ev_m \::\: m\in \IOb
\Setclose.
\end{equation*}

Our next axiom states that the sets of events observed by any two observers are the same. 
\begin{description}
\item[\Ax{AxEv}] All observers coordinatize the same events:
\begin{equation*}
\forall m,k\in \IOb\enskip\forall p\in\Q^d\; \exists q\in\Q^d \quad \ev_m(p)=\ev_k(q).
\end{equation*}
\end{description}

We define the {\bf coordinate-function} of observer $m$, in symbols $Crd_m$, as the inverse of the event-function, i.e., 
\begin{equation*}
\Df{Crd_m}\leteq \ev_m^{-1},
\end{equation*}
where $R^{-1}\leteq\setopen\langle y,x\rangle : \langle x,y\rangle \in
R\setclose$ is the FOL definition of the {\bf inverse} of binary
relation $R$.  Let us note that by this definition,
coordinate-function $Crd_m$ may not be a function (in case $ev_m$ is
not one-to-one). It is only a binary relation.

\begin{conv}\label{crdconv}
Whenever we write $Crd_m(e)$, we mean that there is a unique $q \in
\Q^d$ such that $\ev_m(q)=e$, and this $q$ is denoted by $Crd_m(e)$.
That is, if we talk about the value $Crd_m(e)$, we postulate that it
exists and is unique.
\end{conv}

\noindent
The \df{time of event} $e$ according to observer $m$ is defined as 
\begin{equation*}
\Df{\time_m(e)}\leteq Crd_m(e)_\tau,
\end{equation*}
and the \df{elapsed time} between events $e_1$ and $e_2$ measured by
observer $m$ is defined as
\begin{equation*}
\Df{\time_m(e_1,e_2)}\leteq|\time_m(e_1)-\time_m(e_2)|;
\end{equation*}
$\time_m(e_1,e_2)$ is called the {\it proper time} measured by $m$ between $e_1$ and $e_2$ if $m\in e_1\cap e_2$.
Let us  note that whenever we write $\time_m$, we
assume that the events in its argument have unique coordinates by
Convention~\ref{crdconv}.

The \df{coordinate-domain} of observer $m$, in symbols $Cd_m$, is the set of coordinate points where $m$
observes something, i.e.,
\begin{equation*}
\Df{Cd_m}\leteq\{\, p\in\Q^d\::\: \ev_m(p)\neq\emptyset\,\}.
\end{equation*}
The {\bf worldview transformation} between the coordinate-domains of observers $k$ and $m$ is defined as
\begin{equation*}
\Df{w^k_m}\leteq\Setopen\langle q,p\rangle\in Cd_k\times Cd_m\::\:\ev_k(q)=\ev_m(p)\Setclose.
\end{equation*}
Let us note that worldview transformations are only binary relations by this definition.

\begin{conv}\label{wmkconv}
Whenever we write $w^k_m(q)$, we mean there is a unique $p\in \Q^d$ such that $\langle q,p\rangle \in w^k_m$, and this $p$ is denoted by $w^k_m(q)$.
\end{conv}

Let $\Df{1_t}\leteq\langle 1,0,\ldots,0\rangle$.
The {\bf time-unit vector} of $k$ according to $m$ is defined as
\begin{equation*}
\Df{1^k_m}\leteq w^k_m(1_t)-w^k_m(o).
\end{equation*}

The \df{world-line} of body $b$ according to observer $m$ is defined as
the set of coordinate points where $b$ was observed by $m$, i.e.,
\begin{equation*}
\Df{\wl_m(b)}\leteq\{\, p\in \Q^d \::\: W(m,b,p)\,\}.
\end{equation*}

\begin{description}
\item[\Ax{AxLinTime}] The world-lines of observers are lines and time is elapsing uniformly on them:
\begin{multline*}
\forall m,k \in \IOb\enskip \wl_m(k)=\{\, w^k_m(o)+\lambda\cdot1^k_m \::\: \lambda\in\Q\,\}\, \land \\
\forall p,q \in \wl_m(k) \quad \time_k\big(\ev_m(p),\ev_m(q)\big)\cdot\big|1^k_m\big|=|p-q|.
\end{multline*}
\end{description}
\noindent
Let us collect the axioms introduced so far in an axiom system:
\begin{equation*}
\boxed{\ax{Kinem_0}\leteq\Setopen \ax{AxEOF}, \ax{AxSelf}, \ax{AxLinTime}, \ax{AxEv} \Setclose}
\end{equation*}
Let us note that \ax{Kinem_0} is a general axiom system of kinematics 
 in which no relativistic effect is assumed.
\ax{Kinem_0} is a subtheory of Newtonian and relativistic kinematics.

%%%%%%%%%%%%%%%%%%%%%%%%%%%%%%%%%%%%%%%%%%%%%%%%%%%%%%%%%%%
\section{Geometrical Characterization of TwP}%
%%%%%%%%%%%%%%%%%%%%%%%%%%%%%%%%%%%%%%%%%%%%%%%%%%%%%%%%%%%
\label{thm-sec}

Since the axiom systems we use here deal only with inertial motions of
observers, we formulate the inertial version of TwP, which
is also called clock paradox in the literature.\footnote{This inertial version
  is the one that was formulated by Einstein in his famous 1905 paper,
  see \cite[\S 4]{einstein}.} Logical investigation of the
accelerated version of TwP needs a more complex mathematical
apparatus, see \cite{twp}, \cite[\S 4.3]{mythes}, \cite[\S 7]{myphd}.
We also formulate and characterize variants of TwP: one where the
stay-at-home twin will be the younger one (Anti-TwP) and another where
no differential aging will take place (No-TwP).

To formulate TwP, first we formulate the situations in which it can
occur.  We say that observer $m$ observes observers $a$, $b$ and $c$
in a {\bf twin paradox situation} at events $e$, $e_a$ and $e_c$ iff
$a\in e_a\cap e$, $b\in e_a\cap e_c$, $c\in e\cap e_c$, $b\not\in e$
and $\time_m(e_a)<\time_m(e)<\time_m(e_c)$ or
$\time_m(e_a)>\time_m(e)>\time_m(e_c)$, see Figure~\ref{twpp}.  This
situation is denoted by $\Df{\mtwp_m(\widehat{ac},b)}(e_a,e,e_c)$.

\begin{figure}[h!btp]
\begin{center}
\psfrag{m}{$m$}
\psfrag{a}{$a$}
\psfrag{b}{$b$}
\psfrag{c}{$c$}
\psfrag{p}{$p$}
\psfrag{s}{$s$}
\psfrag{q}{$q$}
\psfrag{0}{$o$}
\psfrag{r}[bl][bl]{$r$}
\psfrag{e}{$e$}
\psfrag{ea}{$e_a$}
\psfrag{ec}{$e_c$}
\psfrag{Ev}[tr][tr]{$Ev$}
\psfrag{Crdm}{$Crd_m$}
\psfrag{1b}[rb][rb]{${b}^\ddag$}
\psfrag{1a}[rb][rb]{${a}^\ddag$}
\psfrag{1c}[rb][rb]{${c}^\ddag$}
\includegraphics[keepaspectratio, width=0.8\textwidth]{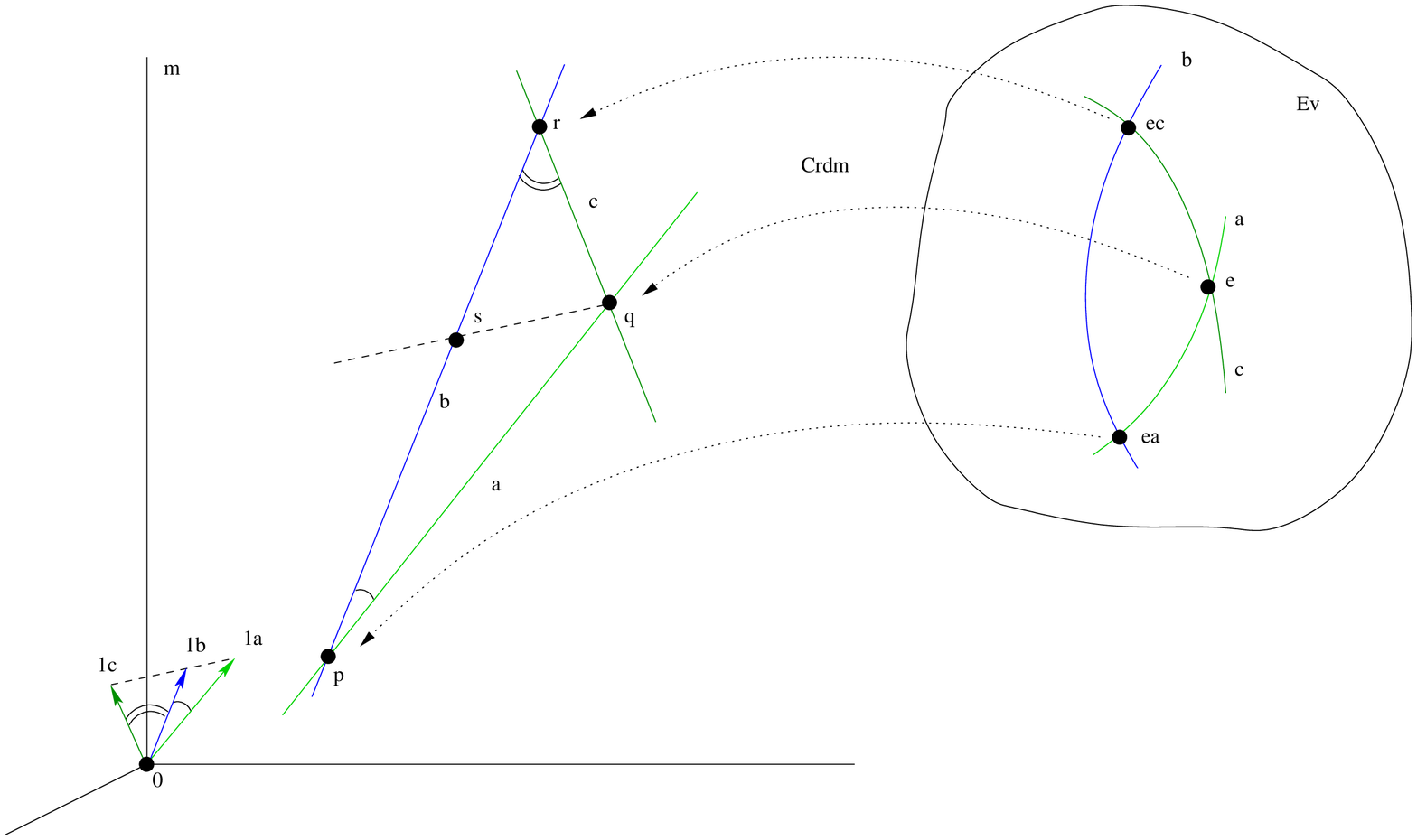}
\caption{\label{twpp} Illustration of relation
  $\mtwp_m(\widehat{ac},b)(e_a,e,e_c)$ and the proof of
  Proposition~\ref{prop-tp}}
\end{center}
\end{figure}

Let $a,b,c\in\IOb$ and $e_a,e,e_b\in Ev$.  Let
$\Df{\time(\widehat{ac}<b)}(e_a,e,e_b)$ be the abbreviation of
$\time_a(e_a,e)+\time_c(e,e_c)<\time_b(e_a,e_c)$.  The definitions of
$\Df{\time(\widehat{ac}=b)}(e_a,e,e_b)$ and
$\Df{\time(\widehat{ac}>b)}(e_a,e,e_b)$ are analogous.  Using this
notation, we can formulate the twin paradox as follows:
\begin{description}
\item[\Ax{TwP}] Every observer $m$ observes the twin paradox in every twin paradox situation:
\begin{multline*}
\forall m,c,a,b\in \IOb\enskip \forall e,e_a,e_c\in Ev_m\\
\mtwp_m(\widehat{ac},b)(e_a,e,e_c) \then \time(\widehat{ac}<b)(e_a,e,e_c).
\end{multline*}
\end{description}
We define \Ax{noTwP} and \Ax{antiTwP} by replacing $<$ by $=$ and $>$ in the formula \ax{TwP}, respectively.

\begin{rem}
For convenience, we quantify over events too.  That does not mean
abandoning our FOL language.  It is just simplifying the
formalization of our axioms.  Instead of events we could speak about
observers and spacetime locations.  For example, instead of $\forall
e\in Ev_m\enskip \phi$ we could write $\forall p\in Cd_m\enskip
\phi[e\!\leadsto\! \ev_m(p)]$, where none of $p_1\ldots p_d$ occurs
free in $\phi$, and $\phi[e\!\leadsto\!  \ev_m(p)]$ is the formula
obtained from $\phi$ by substituting $\ev_m(p)$ for $e$ in all free
occurrences.  Similarly, we can replace $\forall e\in Ev\enskip \phi$
by $\forall m\in\IOb\enskip\forall e\in Ev_m\enskip \phi$.
\end{rem}

\begin{figure}[h!btp]
\begin{center}
\psfrag{p}{$p$}
\psfrag{q1}{$q_3$}
\psfrag{q2}{$q_2$}
\psfrag{q3}{$q_1$}
\psfrag{r}{${}^\ddag r,r$}
\psfrag{p+}{${}^\ddag p$}
\psfrag{o}{$o$}
\psfrag{convex}[bl][bl]{convex}
\psfrag{flat}[l][l]{flat}
\psfrag{concave}[bl][bl]{concave}
\includegraphics[keepaspectratio, width=0.6\textwidth]{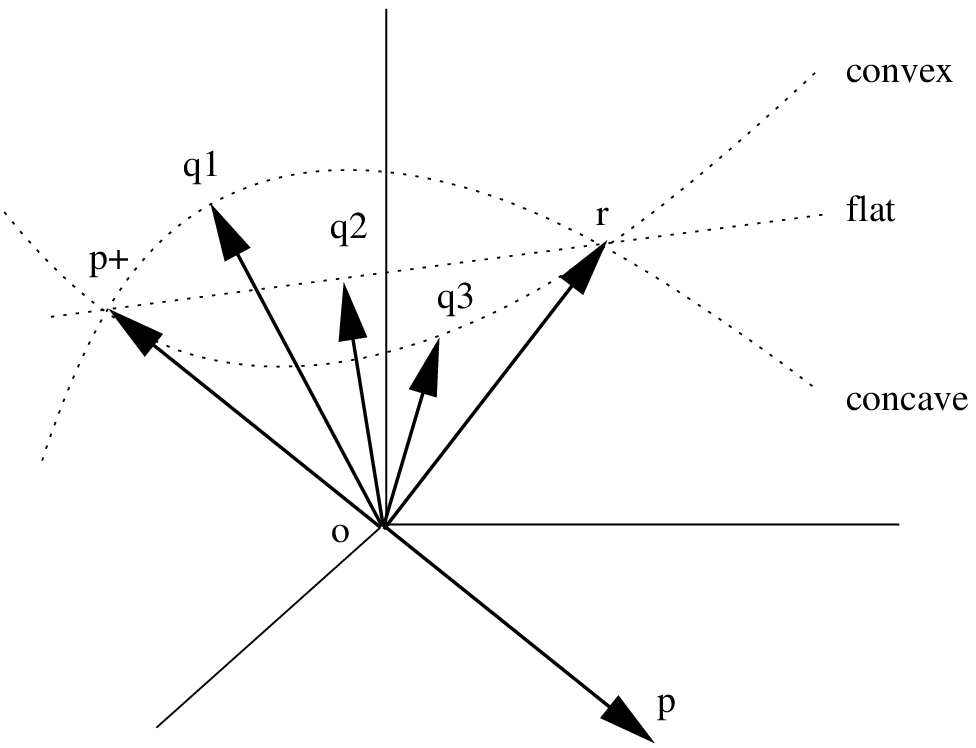}
\caption{\label{fig-conv} Illustration of relations $\convex(p,q_1,r)$, $\Bw(p,q_2,r)$ and $\concave(p,q_3,r)$}
\end{center}
\end{figure}

We say that $q\in\Q^d$ is (strictly) {\bf between} $p\in\Q^d$ and
$r\in\Q^d$ iff there is $\lambda\in Q$ such that $q=\lambda
p+(1-\lambda)r$ and $0<\lambda<1$.  This situation is denoted by
$\Df{\Bw}(p,q,r)$.
Let $p,q,r \in \Q^d$ and $\mu\in\Q$ such that $\Bw(p,\mu q,r)$.  In
this case we use notations $\Df{\convex}(p,q,r)$ and
$\Df{\concave}(p,q,r)$ if $1<\mu$ and $0<\mu<1$, respectively.
For convenience, we introduce the following notation:
\begin{equation*}
\Df{{}^\ddag p}\leteq\left\{
\begin{array}{ll}
\phantom{-}p &\text{ if \;} p_t \ge 0,\\
-p & \mbox{ if\; } p_t < 0.
\end{array}
\right.
\end{equation*}

\begin{prop}\label{prop-tp}
Assume \ax{Kinem_0}.
Let $m$, $a$, $b$, and $c$ be observers and $e$, $e_a$ and $e_c$ events such that $\mtwp_m(\widehat{ac},b)(e_a,e,e_c)$.
Then 
\begin{alignat*}{3}
&\time(\widehat{ac}<b)(e_a,e,e_c) &\quad &\Iff &\quad
  &\convex({}^\ddag1_m^a,{}^\ddag1_m^b,{}^\ddag1_m^c),\\ &\time(\widehat{ac}=b)(e_a,e,e_c)
  & &\Iff &
  &\Bw({}^\ddag1_m^a,{}^\ddag1_m^b,{}^\ddag1_m^c),\\ &\time(\widehat{ac}>b)(e_a,e,e_c)
  & &\Iff & &\concave({}^\ddag1_m^a,{}^\ddag1_m^b,{}^\ddag1_m^c).
\end{alignat*}
\end{prop}

\begin{proof}[Proof]
Let $m$, $a$, $b$, and $c$ be observers and $e$, $e_a$ and $e_c$
events such that $\mtwp_m(\widehat{ac},b)(e_a,e,e_c)$.  Let us
abbreviate time-unit vectors ${}^\ddag 1^k_m$ to $k^\ddag$ throughout
this proof.  Let $p=Crd_m(e_a)$, $q=Crd_m(e)$ and $r=Crd_m(e_c)$.  We
have that $p\neq r$ since $p_\tau<r_\tau$ or $r_\tau<p_\tau$.
Therefore, by \ax{AxLinTime}, the triangle $pqr$ is nondegenerate
since $p,r\in \wl_m(b)$ but $q\not\in \wl_m(b)$.  Let us first show that
$b$ measures the same length of time between $e_a$ and $e_c$ as $a$
and $c$ together if $\Bw(a^\ddag,b^\ddag,c^\ddag)$ holds.  Let $s$
be the intersection of line $pr$ and the line parallel to $a^\ddag
c^\ddag$ through $q$, see Figure~\ref{twpp}.  Since $\Bw(a^\ddag,b^\ddag,c^\ddag)$ holds, the triangles
$oa^\ddag b^\ddag$ and $pqs$ are similar; and the triangles $ob^\ddag
c^\ddag$ and $rsq$ are similar.  Thus
\begin{equation*}
\frac{|p-q|}{|a^\ddag|}=\frac{|p-s|}{|b^\ddag|} \;\text{ and }\; \frac{|q-r|}{|c^\ddag|}=\frac{|s-r|}{|b^\ddag|}
\end{equation*}
hold. 
From which, by \ax{AxLinTime}, it follows that
\begin{multline*}
\Big|\time_a(e_a,e)\Big|+\Big|\time_c(e,e_c)\Big| = \frac{|p-q|}{|a^\ddag|}+\frac{|q-r|}{|c^\ddag|}\\
=\frac{|p-s|+|s-r|}{|b^\ddag|}=\frac{|r-p|}{|b^\ddag|}= \Big|\time_c(e_a,e_c)\Big|.
\end{multline*}
Hence $\time(\widehat{ac}=b)(e_a,e,e_c)$ holds if
$\Bw(a^\ddag,b^\ddag,c^\ddag)$.  By \ax{AxLinTime}, $b$ measures more
(less) time between $e_a$ and $e_c$ iff his time-unit vector is
shorter (longer).  Thus we get that $\time(\widehat{ac}<b)(e_a,e,e_c)$
holds if $\convex(a^\ddag,b^\ddag,c^\ddag)$, and
$\time(\widehat{ac}>b)(e_a,e,e_c)$ holds if
$\concave(a^\ddag,b^\ddag,c^\ddag)$.  The converse implications also
hold since one of the relations $\convex$, $\Bw$ and $\concave$ holds
for $a^\ddag$, $b^\ddag$ and $c^\ddag$, and only one of the relations
$\time(\widehat{ac}<b)$, $\time(\widehat{ac}=b)$ and
$\time(\widehat{ac}>b)$ can hold for events $e_a$, $e$ and $e_c$.
This completes the proof.
\end{proof}

A set $H\subseteq \Q^d$ is called {\bf convex} iff $\convex(p,q,r)$
for all $p,q,r\in H$ for which there is $\mu\in Q$ such that $\Bw(p,\mu q,
r)$.  We call $H$ {\bf flat} or {\bf concave} if $\convex(p,q,r)$ is
replaced by $\Bw(q,r,p)$ or $\concave(r,p,q)$, respectively. 
\begin{rem}\label{rem-conv}
If there are no $p,q,r\in H$ for which there is a $\mu\in \Q^+$
such that $\Bw(p,\mu q,r)$ holds, then $H$ is convex, flat and
concave at the same time. To avoid these undesired situations, let us
call $H$ \df{nontrivial} if there are $p,q,r\in H$ such that
$\Bw(p,\mu q,r)$ holds for a $\mu\in \Q^+$. By the respective
definitions, it is easy to see that any nontrivial convex (flat,
concave) set intersects a halfline at most once.
\end{rem}
\noindent
Let us define the {\bf Minkowski sphere} here as
$\Df{MS^\ddag_m}\leteq\big\{\, {}^\ddag 1^k_m \::\: k\in\IOb\,\big\}.$

\begin{rem}
Convexity as used here is not far from convexity as understood in
geometry or in the case of functions.  For example, in the models of
$\ax{Kinem_0}+\ax{AxThExp^+}$ or $\ax{SpecRel^-_\d}+\ax{AxThExp}$ (see next sections) the
Minkowski Sphere $MS^\ddag_m$ is convex in our sense iff the set of
points above it $\setopen p\in\Q^d \,:\, \exists q\in
MS^\ddag_m\enskip p_\tau\ge q_\tau\setclose$ is convex in the
geometrical sense.
\end{rem}

\begin{rem}\label{rem-convMS}
By Remark \ref{rem-conv}, if $MS^\ddag_m$ is a nontrivial convex
(flat, concave) set, it intersects a line at most once.
\end{rem}
\noindent
The following is a corollary of Proposition~\ref{prop-tp}.
\begin{cor} 
\label{cor-twp}
Assume \ax{Kinem_0}.
Then
\begin{alignat*}{3}
&\forall m\in\IOb\enskip MS^\ddag_m \text{ is convex} &\enskip &\Then &\enskip & \ax{TwP}, \\
&\forall m\in\IOb\enskip MS^\ddag_m \text{ is flat} & &\Then & &\ax{noTwP}, \\
&\forall m\in\IOb\enskip MS^\ddag_m \text{ is concave}& &\Then & &\ax{antiTwP}.
\end{alignat*}
\end{cor}

The implications in Corollary \ref{cor-twp} cannot be reversed
because there may be observers that are not part of any twin paradox
situation.  We can resolve this problem by using the following axiom to
shift observers in order to create twin paradox situations.
\begin{description}
\item[\Ax{AxShift}] If an observer observes another observer with a
  certain time-unit vector, it also observes still another observer,
  with the same time-unit vector, at each coordinate point of its
  coordinate domain:
\begin{equation*}
\forall m,k\in \IOb\enskip \forall p\in Cd_m\; \exists h \in\IOb \quad h\in \ev_m(p) \land 1^k_m=1^h_m.
\end{equation*} 
\end{description}
Axiom \ax{AxShift} postulates the existence of some observers. Since
by observes (bodies) we mean potential observers (potential bodies)
these kinds of assumptions are quite natural, see also axioms
\ax{AxThExp^+}, \ax{AxThExp^*} and \ax{AxThExp} on pages
\pageref{axthexp+} and \pageref{axthexp}.  Now we can reverse the
implications of Corollary \ref{cor-twp}.
\begin{thm}\label{thm-twp}
Assume \ax{Kinem_0} and \ax{AxShift}. Then
\begin{alignat*}{3}
&\ax{TwP} &\enskip &\Iff &\enskip & \forall m\in\IOb\enskip MS^\ddag_m \text{ is convex,}\\
&\ax{noTwP} & &\Iff & &\forall m\in\IOb\enskip MS^\ddag_m \text{ is flat,}\\
&\ax{antiTwP} & &\Iff & &\forall m\in\IOb\enskip MS^\ddag_m \text{ is concave}.
\end{alignat*}
\end{thm}

\begin{proof}[Proof]
By Corollary~\ref{cor-twp}, we have to prove the ``$\Longrightarrow$''
part only.  For that, let us take three points $a'$, $b'$ and $c'$
from $MS^\ddag_m$ for which there is $\mu\in Q$ satisfying
$\Bw({}^\ddag a',\mu b', {}^\ddag c')$. If there are no such points,
$MS^\ddag_m$ is convex, flat and concave at the same time, see Remark
\ref{rem-conv}. Otherwise, by \ax{AxShift} there are
observers $a$, $b$ and $c$ in a twin paradox situation such that
$1^a_m=a'$, $1^b_m=b'$ and $1^c_m=c'$.  Thus from
Proposition~\ref{prop-tp} we get that $MS^\ddag_m$ has the desired
property.
\end{proof}

In the sections below we will use the following concept.  Let $\Sigma$
and $\Gamma$ be sets of formulas, and let $\varphi$ and $\psi$ be
formulas of our language.  Then $\Sigma$ {\bf logically implies}
$\varphi$, in symbols $\Sigma\Df{\models}\varphi$, iff $\varphi$ is
true in every model of $\Sigma$.  To simplify our notations, we use
the plus sign between formulas and sets of formulas in the following
way: $\Sigma+\Gamma\leteq \Sigma\cup\Gamma$,
$\varphi+\psi\leteq\setopen\varphi,\psi\setclose$ and
$\Sigma+\varphi\leteq\Sigma\cup\setopen\varphi\setclose$.

\begin{rem}
Let us note that the fewer axioms $\Sigma$ contains, the stronger the
logical implication $\Sigma\models\varphi$ is, and similarly the more
axioms $\Sigma$ contains the stronger the counterexample
$\Sigma\not\models\varphi$ is.
\end{rem}

\begin{rem}
By G\"odel's completeness theorem, all the theorems of this paper remain
valid if we replace the relation of logical consequence ($\models$) by
the deducibility relation of FOL ($\vdash$).
\end{rem}

%%%%%%%%%%%%%%%%%%%%%%%%%%%%%%%%%%%%%%%%%%
\section{Consequences for Newtonian kinematics}
%%%%%%%%%%%%%%%%%%%%%%%%%%%%%%%%%%%%%%%%%%
\label{conk-sec}

Let us investigate the logical connection between No-TwP and the Newtonian assumption on the absoluteness of time.

\begin{description}
\item[\Ax{AbsTime}] Observers measure the same time elapsing between events:
\begin{equation*}
\forall m,k\in \IOb\enskip\forall e_1,e_2\in Ev\quad \time_m(e_1,e_2)=\time_k(e_1,e_2).
\end{equation*}
\end{description}

To strengthen our axiom system, we introduce two axioms that ensure the
existence of several observers.
\begin{description}
\item[\Ax{AxThExp^+}]\label{axthexp+} Observers can move in any direction at any finite speed:
\begin{equation*}
\forall m\in \IOb\enskip\forall p,q\in Q^d \quad p_\tau\neq q_\tau 
\then \exists k\in\IOb\quad k\in \ev_m(p)\cap \ev_m(q).
\end{equation*}
\end{description}
This axiom as well as its variants (\ax{AxThExp^*} below and
\ax{AxThExp} on page \pageref{axthexp}) are closely related to the
assumptions of homogeneity and isotropy of space since in some respect
they say that there is no difference between the different points and
directions in space. A more experimental version of axiom
\ax{AxThExp^+} is the following:
\begin{description}
\item[\Ax{AxThExp^*}]\label{axthex*} Observers
  can move in any direction at a speed which is arbitrarily close to any finite speed:
\begin{multline*}
\forall m\in \IOb\enskip \forall p,q\in \Q^d\enskip\forall
\varepsilon\in \Q^+ \quad p_\tau\neq q_\tau \\\then
\exists k\in\IOb\enskip \exists {q}'\in \Q^d\quad
|q-{q}'|<\varepsilon \land k\in \ev_m(p)\cap \ev_m({q}').
\end{multline*}
\end{description}
Since the accuracy of an experiment is finite and we can make only
finitely many experiments, axiom \ax{AxThExp^*} is a more plausible
assumption than \ax{AxThExp^+} from empirical point of view.

By the following theorem, \ax{noTwP} logically implies \ax{AbsTime} if
\ax{AxThExp^+} (and some auxiliary axioms) are assumed; however, if we
assume the more experimental axiom \ax{AxThExp^*} instead of
\ax{AxThExp^+}, \ax{AbsTime} does not follow from \ax{noTwP}, which is
an astonishing fact since it means that without the strong theoretical
assumption \ax{AxThExp^+} we would not be able to conclude that time
is absolute in the Newtonian sense even if there were no twin paradox
in our world.

\begin{thm} 
\label{thm-univtime}
\begin{align}
\label{eq-notwp1}
\ax{AxEOF}+\ax{AbsTime}&\models \ax{noTwP}, \text{ and}\\
\label{eq-notwp1.5}
\ax{Kinem_0}+\ax{AxShift}+\ax{AxThExp^+}+\ax{noTwP}&\models \ax{AbsTime}, \text{ but}\\
\label{eq-notwp2}
\ax{Kinem_0}+\ax{AxShift}+\ax{AxThExp^*}+\ax{noTwP}&\not\models \ax{AbsTime}.
\end{align}
\end{thm}

\begin{proof}[Proof]
Item \eqref{eq-notwp1} is obvious.

To prove \eqref{eq-notwp1.5}, let us note that $MS^\ddag_m$ is flat by
Theorem \ref{thm-twp} since \ax{Kinem_0}, \ax{AxShift} and \ax{noTwP}
are assumed. So $MS^\ddag_m$ is a subset of a hyperplane. By axiom
\ax{AxThExp^+}, $MS^\ddag_m$ intersects any nonhorizontal line. If the
hyperplane containing $MS^\ddag_m$ were not horizontal, there would be
nonhorizontal lines parallel to it. Therefore $MS^\ddag_m$ has to be a
subset of a horizontal hyperplane.  If $MS^\ddag_m$ were a porper
subset of this hyperplane, there would be nonhorizontal lines not
intersecting it.  So $MS^\ddag_m$ has to be a horizontal hyperplane
containing $\langle 1,0,\ldots,0 \rangle=1^m_m$. Hence the time
components of time-unit vectors are the same for every observer. So
\ax{AbsTime} follows from the assumptions.

To prove \eqref{eq-notwp2}, we construct a model in which
\ax{Kinem_0}, \ax{AxShift}, \ax{AxThExp^*} and \ax{noTwP} hold, but
\ax{AbsTime} does not.  Let $\langle \Q;+,\cdot,<\rangle$ be any
Euclidean ordered field.  Let $B\leteq \Q^d\times\Q^d$.  Let
$\IOb\leteq \setopen \langle p,q\rangle\in\B\,:\, p_\tau\neq
q_\tau \land p_\tau-q_\tau\neq p_2-q_2\setclose$.  Let
$MS^\ddag_{\langle 1,0\rangle}\leteq \{\, x\in\Q^d\,:\,x_\tau-x_2=1 \land x_\tau>0\,\}.$
Let $W(\langle1,0\rangle,\langle p,q\rangle,r)$ hold iff $r$
is in the line  through $p$ and $q$.  Now the worldview relation is given for
observer $\langle1,0\rangle$.  For any other  observer
$\langle p,q\rangle$, let $w^{\langle p,q\rangle}_{\langle1,0\rangle}$
be an affine transformation that takes $o$ to $p$ while its linear
part takes $1_t$ to $MS^\ddag_{\langle 1,0\rangle}\cap\setopen
\lambda(p-q):\lambda\in \Q\setclose$, and leaves  the other basis
vectors fixed.  From these worldview transformations, it is easy to define
the worldview relations of other observers, hence our
model is given.  It is not difficult to see that \ax{Kinem_0},
\ax{AxShift} and \ax{AxThExp^*} are true in this model.  Since
$MS^\ddag_{\langle 1,0\rangle}$ is flat and the worldview
transformations are affine ones, it is clear that $MS^\ddag_m$ is flat
for all $m\in \IOb$.  Hence \ax{noTwP} is also true in this model by
Corollary~\ref{cor-twp}.  It is easy to see that \ax{AbsTime} implies
that $(1^k_m)_\tau=\pm1$ for all $m,k\in\IOb$.  Hence \ax{AbsTime} is
not true in this model, as we claimed.  
\end{proof}

%%%%%%%%%%%%%%%%%%%%%%%%%%%%%%%%%%%%%%%%%%
\section{Consequences for special relativity theory}
%%%%%%%%%%%%%%%%%%%%%%%%%%%%%%%%%%%%%%%%%%
\label{consr-sec}

Now we are going to investigate the consequences of Theorem
\ref{thm-twp} for special relativity.  To do so, let us extend our
language by a new unary relation $\Ph$ on $\B$ for {\bf photons}
(light signals) and formulate an axiom on the constancy of the speed
of light.  For convenience, this speed is chosen to be 1.

\begin{description}
\item[\Ax{AxPh}] For every observer, there is a photon through two
coordinate points $p$ and $q$ iff the slope of $p-q$ is $1$:
\begin{multline*}
\forall m\in \IOb\enskip \forall p,q\in \Q^d\quad |p_\sigma-q_\sigma|=|p_\tau-q_\tau| \\ 
 \iff \Ph\cap \ev_m(p)\cap \ev_m(q)\ne\emptyset.
\end{multline*}
\end{description}
Let us also introduce a symmetry axiom.
\begin{description}\label{psymd}
\item[\Ax{AxSymDist}]If events $e_1$ and $e_2$ are simultaneous for
  both the observers $m$ and $k$, then $m$ and $k$ agree as to the
  spatial distance between $e_1$ and $e_2$:
\begin{multline*}
\forall m,k\in \IOb\enskip \forall e_1,e_2 \in Ev \quad\time_m(e_1,e_2)=\time_k(e_1, e_2)=0 \\
\then \dist_m(e_1,e_2)=\dist_k(e_1,e_2),
\end{multline*}
\end{description}
where the \df{spatial distance} between events $e_1$ and $e_2$
according to observer $m$, in symbols $\Df{\dist_m}(e_1,e_2)$, is formulated as $|Crd_m(e_1)_\sigma-Crd_m(e_2)_\sigma|$.

\noindent
Let us introduce the following axiom system:
\begin{equation*}\label{page-sr}
\boxed{\ax{SpecRel_\d}\leteq \Setopen \ax{AxEOF},\ax{AxSelf}, \ax{AxPh}, \ax{AxEv},\ax{AxSymDist} \Setclose}
\end{equation*}
Now we have a FOL axiom system of special relativity for each
natural number $d\ge2$.

To state the Alexandrov-Zeeman theorem
generalized for fields, we need a definition.
 A map
$\tilde\varphi:\Q^d\rightarrow \Q^d$ is called a {\bf
  field-automorphism-induced} iff there is an automorphism
$\varphi$ of the field $\langle\Q,\cdot,+\rangle$ such that
$\tilde\varphi(p)=\langle \varphi(p_1),\ldots,\varphi(p_d)\rangle$ for
every $p\in\Q^d$. 

\begin{thm}[Alexandrov-Zeeman]\label{thm-az}
Let $F$ be a field and $d\ge3$.  Every bijection from $F^d$ to $F^d$
that transforms lines of slope 1 to lines of slope 1 is a Poincar\'e
transformation composed by a dilation and a field-automorphism-induced
map.
\end{thm}
For the proof of Theorem \ref{thm-az}, see \cite{vro},
\cite{VKK}. From this theorem we derive that the worldview
transformations between observers are Poincar\'e ones in the models of
\ax{SpecRel_\d} if $\d\ge3$, cf.\ \cite[Theorem 11.11,
  p.641]{logst}. This fact justifies our calling \ax{SpecRel_\d} an
axiom system of special relativity.

\begin{thm}
\label{thm-poi}
Let $d\ge 3$.
Let $m,k\in\IOb$.
Then
\begin{enumerate}
\item\label{item-poi1} if \ax{AxEOF}, \ax{AxPh} and \ax{AxEv} are assumed,
  $w^k_m$ is a Poincar\'e transformation composed by a dilation $D$
  and a field-automorphism-induced map $\tilde\varphi$;
\item\label{item-poi2} if \ax{AxEOF}, \ax{AxPh}, \ax{AxEv} and \ax{AxSymDist}
  are assumed, $w^k_m$ is a Poincar\'e transformation.
\end{enumerate}
\end{thm}

\begin{proof}[On the proof]
It is not difficult to see that \ax{AxPh} and \ax{AxEv} imply that
$w^k_m$ is a bijection from $\Q^d$ to $\Q^d$ that preserves lines of
slope 1, see, e.g., \cite[Proposition 3.1.3]{myphd}.  Hence Item
\eqref{item-poi1} is a consequence of Theorem \ref{thm-az}.

Now let us see why Item \eqref{item-poi2} is true.  By Item
\eqref{item-poi1}, it is easy to see that there is a line $l$ such
that both $l$ and its $w^k_m$ image are orthogonal to the time-axis.
Thus by \ax{AxSymDist}, $w^k_m$ restricted to $l$ is distance
preserving.  Consequently, both the dilation $D$ and the
field-automorphism-induced map $\tilde\varphi$ in Item
\eqref{item-poi2} have to be the identity map.  Hence $w^k_m$ is a
Poincar\'e transformation.
\end{proof}

\noindent
Let us now formulate another famous prediction of relativity.
\begin{description}
\item[\Ax{SlowTime}] Relatively moving observers' clocks slow down:
\begin{equation*}
\forall m,k\in \IOb\quad \wl_m(k)\neq \wl_m(m) \then \big|(1^k_m)_\tau\big|>1.
\end{equation*}
\end{description}

To investigate the logical connection between \ax{SlowTime} and \ax{TwP}, let us also introduce a weakened axiom system of special relativity:
\begin{equation*}
\boxed{\ax{SpecRel^-_\d}\leteq\Setopen \ax{AxEOF}, \ax{AxSelf}, \ax{AxPh}, \ax{AxEv} \Setclose}
\end{equation*}

Let us note that if $d\ge3$, \ax{SpecRel^-_\d} is strong enough to prove
the most important predictions of special relativity, such as that
moving clocks get out of synchronism, see, e.g., \cite{AMNsamples}.
At the same time, \ax{SpecRel^-_\d} is weak enough not to prove every
prediction of special relativity.  For example, it does not entail \ax{TwP}
or \ax{SlowTime}.  Thus it is possible to compare these predictions
within \ax{SpecRel^-_\d}.  

To prove a theorem about the logical connection between \ax{SlowTime}
and \ax{TwP}, we need the following lemma, which states that the fact
that three observers are in a twin paradox situation does not depend on
the observer that watches them.

\begin{lem}
\label{lem-mtwp}
Let $d\ge3$.
Assume \ax{AxEOF}, \ax{AxPh}, \ax{AxEv} and \ax{AxLinTime}.
Let $m,a,b,c\in\IOb$ and let $e_a,e,e_b\in Ev$.
Then
\begin{equation*}
\mtwp_m(\widehat{ac},b)(e_a,e,e_c)\Iff \mtwp_b(\widehat{ac},b)(e_a,e,e_c).
\end{equation*}
\end{lem}

\begin{proof}[Proof]
By (1) of Theorem~\ref{thm-poi}, $w^b_m$ is a composition of a
Poincar\'e transformation, a dilation and a field-automorphism-induced
map since \ax{AxEOF}, \ax{AxPh} and \ax{AxEv} are assumed.  By
\ax{AxLinTime}, the field-automorphism is trivial.  Hence $\time_m(e)$
is between $\time_m(e_a)$ and $\time_m(e_c)$ iff $\time_b(e)$ is
between $\time_b(e_a)$ and $\time_b(e_c)$.  This completes the proof
since the other parts of the definition of relation $\mtwp$ do not
depend on observers $m$ and $b$.
\end{proof}

We cannot consistently extend our theory \ax{SpecRel^-_\d} by axiom \ax{AxThExp^+} since
\ax{SpecRel^-_\d} implies the impossibility of faster than light motion
of observers if $d\ge3$, see, e.g., \cite{AMNsamples}. That is why we have
to weaken this axiom.

\begin{description}
\item[\Ax{AxThExp}]\label{axthexp} Observers can move in any direction at any speed
  slower than 1, i.e., less than the speed of light:
\begin{multline*}
\forall m\in \IOb\enskip\forall p,q\in Q^d \quad |p_\sigma-q_\sigma|<|p_\tau-q_\tau| \\
\then \exists k\in\IOb\quad k\in \ev_m(p)\cap \ev_m(q).
\end{multline*}
\end{description}

\noindent
The following theorem shows that \ax{SlowTime} is logically stronger than \ax{TwP}.

\begin{thm}
\label{thm-slowtime} 
Let $d\ge3$. Then
\begin{align}
\label{eq-slowtime1}
\ax{SpecRel^-_\d}+\ax{AxLinTime}+\ax{SlowTime}&\models \ax{TwP},\text{ but}\\
\label{eq-slowtime2}
\ax{SpecRel^-_\d}+ \ax{AxShift}+ \ax{AxLinTime}+\ax{AxThExp}+\ax{TwP}&\not\models \ax{SlowTime}.
\end{align}
\end{thm}

\begin{proof}[Proof]
Item \eqref{eq-slowtime1} is clear by Lemma~\ref{lem-mtwp}.

To prove Item \eqref{eq-slowtime2}, let us construct a model in which
\ax{SpecRel^-_\d}, \ax{AxShift}, \ax{AxLinTime}, \ax{AxThExp} and
\ax{TwP} hold, but \ax{SlowTime} does not. Let $\langle
\Q;+,\cdot,<\rangle$ be any Euclidean ordered field. Let $B\leteq
\Q^d\times\Q^d$. Let $\IOb\leteq \setopen \langle
p,q\rangle\in\B\setmid
|p_\sigma-q_\sigma|<|p_\tau-q_\tau|\setclose$. It is easy to see that
there is a nontrivial convex subset $M$ of $\Q^d$ such that $1_t\in M$
and $|p_\tau|<1$ for some $p\in M$. Let $MS^\ddag_{\langle
  1,0\rangle}$ be such a convex subset of $\Q^d$. Let
$W(\langle1,0\rangle,\langle p,q\rangle,r)$ hold iff $r$ is in
the line through $p$ and $q$. Now the worldview relation is given for observer
$\langle1,0\rangle$. By Remark \ref{rem-convMS}, $MS^\ddag_{\langle
  1,0\rangle}$ intersects a line at most once. For any other observer
$\langle p,q\rangle$, let $w^{\langle p,q\rangle}_{\langle1,0\rangle}$
be such a composition of a Lorentz transformation, a dilation and a
translation which takes $o$ to $p$ while its linear part takes $1_t$
to the unique element of $MS^\ddag_{\langle 1,0\rangle}\cap\setopen
\lambda(p-q):\lambda\in \Q\setclose$, and leaves the other basis
vectors fixed. It is easy to see that there is such a transformation. From
these worldview transformations, it is easy to define the worldview
relations of the other observers. So the model is
given. It is not difficult to see that \ax{SpecRel^-_\d}, \ax{AxShift},
\ax{AxLinTime} and \ax{AxThExp} are true in this model. Since
$MS^\ddag_{\langle 1,0\rangle}$ is convex and the worldview
transformations are affine ones, it is clear that $MS^\ddag_m$ is
convex for all $m\in \IOb$. Hence \ax{TwP} is also true in this model
by Corollary~\ref{cor-twp}. It is clear that \ax{SlowTime} is not true
in this model since there is a $p\in MS^\ddag_{\langle 1,0\rangle}$
such that $|p_\tau|<1$ (i.e., there is $k\in\IOb$ such that
$|(1^k_{\langle 1,0\rangle})_\tau|<1$); and that completes the proof.
\end{proof}

Like the similar results of \cite{mytdk} and \cite{mythes}, the
following theorem answers Question 4.2.17 of
Andr\'eka--Madar\'asz--N\'emeti \cite{pezsgo}.  It shows that \ax{TwP}
is logically weaker than the symmetry axiom of \ax{SpecRel_\d}.
\begin{thm}
\label{thm-simdist}
Let $d\ge3$. 
Then
\begin{align}
\label{eq-simdist1}
 \ax{SpecRel^-_\d}+\ax{AxSymDist}&\models \ax{TwP}, \text{ but}\\
\label{eq-simdist2}
 \ax{SpecRel^-_\d}+ \ax{AxShift}+\ax{AxLinTime} +\ax{AxThExp}+\ax{TwP}&\not\models\ax{AxSymDist}.
\end{align}
\end{thm}

\begin{proof}[Proof]
By \eqref{item-poi2} of Theorem~\ref{thm-poi}, \ax{SpecRel^-_\d} and
\ax{AxSymDist} imply that $w^k_m$ is a Poincar\'e transformation for
all $m,k\in\IOb$.  Hence
$  MS^\ddag_m\subseteq\{\, p\in \Q^d\::\: p_\tau^2- {|p_\sigma|}^2 =1 \land p_\tau>0\,\}.$
Consequently, $MS^\ddag_m$ is convex. 
So by Corollary~\ref{cor-twp}, \ax{TwP} follows from \ax{SpecRel^-_\d} and \ax{AxSymDist}.

Since \ax{SpecRel^-_\d} and \ax{AxSymDist} imply \ax{SlowTime} if $d\ge3$, Item \eqref{eq-simdist2} follows from Theorem~\ref{thm-slowtime}.
\end{proof}

It is interesting that \ax{AxSymDist} and \ax{SlowTime} are equivalent
in the models of \ax{SpecRel^-_\d} (and some auxiliary axioms) if the
quantity part is the field of real numbers.  However, that the
quantity part is the field of real numbers cannot be formulated in any
FOL language of spacetime theories.  Consequently, nor can
Theorem~\ref{thm-eqv}, so it cannot be formulated and proved within
our FOL frame either.
\begin{thm}
\label{thm-eqv}
Let $\d\ge3$. Assume \ax{SpecRel^-_\d}, \ax{AxThExp}, \ax{AxLinTime}, \ax{AxShift}, and that $\Q$ is the field of real numbers.
Then
\begin{align}
 \ax{SlowTime}\Iff\ax{AxSymDist}.
\end{align}
\end{thm}
For proof of Theorem~\ref{thm-eqv}, see \cite[\S 3]{mythes}.  This
theorem is interesting because it shows that assuming only that all
 moving clocks slow down to some degree implies the exact ratio of
the slowing down of moving clocks (since if $\d\ge3$,
$\ax{SpecRel^-_\d}+\ax{AxSymDist}$ implies that the worldview
transformations are Poincar\'e ones, see Theorem~\ref{thm-poi}).
\begin{que}
Does Theorem~\ref{thm-eqv} retain its validity if the assumption that
$\Q$ is the field of real numbers is removed? If not, is it still
possible to replace it by a FOL assumption, e.g., by the axiom schema
of continuity used in \cite{twp}, \cite{logtw}, \cite[\S 7.2]{myphd}?
\end{que}

\section{Concluding remarks}
We have seen that (the inertial version of) TwP can be
characterized geometrically within a general axiom system of
kinematics.  We have also seen some surprising consequences of this
characterization; in particular, that TwP is logically weaker than axiom
\ax{AxSymDist} of special relativity as well as the
assumption of the slowing down of moving clocks.  A future task is to
explore the logical connections between other assumptions and
predictions of relativity theories.  For example, in \cite{twp},
\cite{mythes}, \cite[\S 6]{myphd} \ax{SpecRel_\d} is extended to an
axiom system \ax{AccRel} logically implying the accelerated version of
TwP, but the natural question below, raised by
Theorem~\ref{thm-simdist}, has not been answered yet.
\begin{que} 
Is it possible to weaken \ax{AxSymDist} to \ax{TwP} in \ax{AccRel}
(see, e.g., \cite{myphd}) without losing the accelerated version of
TwP as a consequence? See \cite[Question 3.8]{twp} and \cite[Question
  4.5.6]{myphd}.\end{que}

\section*{ACKNOWLEDGMENTS}
I wish to express my heartfelt thanks to Hajnal Andr\'eka, Judit
X.\ Madar\'asz and Istv\'an N\'emeti for the invaluable inspiration
and guidance I received from them for my work. I am also grateful to
Mike Stannett for his many helpful comments and suggestions.  My
thanks also go to Ram\'on Horv\'ath and Zal\'an Gyenis for our
interesting discussions on the subject.

Research supported by the Hungarian National Foundation for scientific
research grant T73601.

\bigskip\noindent
Alfr\'ed R\'enyi Institute of Mathematics\\
of the Hungarian Academy of Sciences\\
Budapest P.O.Box 127, H-1364, Hungary\\
turms{@}renyi.hu.

\medskip\noindent
Zr{\'i}nyi Mikl{\'o}s University of National Defence\\
Budapest P.O.Box 12, H-1456, Hungary.
\end{document}